\documentclass[a4paper,onecolumn]{IEEEtran}

\usepackage{amsfonts,amssymb,amsmath,amsthm,bm,mathtools}
\usepackage{latexsym,url,color,enumerate,enumitem}
\usepackage{array,amstext}
\usepackage{calc}
\usepackage{wrapfig}
\usepackage{graphicx}
\usepackage{diagbox}
\graphicspath{{figs/}{../figs/}}
\newtheorem{theorem}{Theorem}%[section]
\newtheorem{proposition}[theorem]{Proposition}
\newtheorem{lemma}[theorem]{Lemma}
\newtheorem{corollary}[theorem]{Corollary}
\newtheorem{definition}[theorem]{Definition}
\newtheorem{notation}[theorem]{Notation}
\newtheorem{example}[theorem]{Example}
\newtheorem{remark}[theorem]{Remark}
\newcommand{\ZZ}{\mathbb{Z}}%Integers
\newcommand{\NN}{\mathbb{N}}%Natural numbers
\newcommand{\FF}{\mathbb{F}}%Finite field
\newcommand{\BW}{\textup{BW}}

\newcommand{\wt}{\textup{wt}}

\newcommand{\dist}{\textup{dist}}
\newcommand\nc\newcommand
\nc\bfa{{\boldsymbol a}}\nc\bfA{{\bf A}}\nc\cA{{\mathcal A}}
\nc\bfb{{\boldsymbol b}}\nc\bfB{{\bf B}}\nc\cB{{\mathcal B}}
\nc\bfc{{\boldsymbol c}}\nc\bfC{{\bf C}}\nc\cC{{\mathcal C}}
\nc\bfd{{\boldsymbol d}}\nc\bfD{{\bf D}}\nc\cD{{\mathcal D}}
\nc\bfe{{\boldsymbol e}}\nc\bfE{{\bf E}}\nc\cE{{\mathcal E}}
\nc\bff{{\boldsymbol f}}\nc\bfF{{\bf F}}\nc\cF{{\mathcal F}}
\nc\bfg{{\boldsymbol g}}\nc\bfG{{\bf G}}\nc\cG{{\mathcal G}}
\nc\bfh{{\boldsymbol h}}\nc\bfH{{\bf H}}\nc\cH{{\mathcal H}}
\nc\bfi{{\boldsymbol i}}\nc\bfI{{\bf I}}\nc\cI{{\mathcal I}}
\nc\bfj{{\boldsymbol j}}\nc\bfJ{{\bf J}}\nc\cJ{{\mathcal J}}
\nc\bfk{{\boldsymbol k}}\nc\bfK{{\bf K}}\nc\cK{{\mathcal K}}
\nc\bfl{{\boldsymbol l}}\nc\bfL{{\bf L}}\nc\cL{{\mathcal L}}
\nc\bfm{{\boldsymbol m}}\nc\bfM{{\bf M}}\nc\cM{{\mathcal M}}
\nc\bfn{{\boldsymbol n}}\nc\bfN{{\bf N}}\nc\cN{{\mathcal N}}
\nc\bfo{{\boldsymbol o}}\nc\bfO{{\bf O}}\nc\cO{{\mathcal O}}
\nc\bfp{{\boldsymbol p}}\nc\bfP{{\bf P}}\nc\cP{{\mathcal P}}
\nc\bfq{{\boldsymbol q}}\nc\bfQ{{\bf Q}}\nc\cQ{{\mathcal Q}}
\nc\bfr{{\boldsymbol r}}\nc\bfR{{\bf R}}\nc\cR{{\mathcal R}}
\nc\bfs{{\boldsymbol s}}\nc\bfS{{\bf S}}\nc\cS{{\mathcal S}}
\nc\bft{{\boldsymbol t}}\nc\bfT{{\bf T}}\nc\cT{{\mathcal T}}
\nc\bfu{{\boldsymbol u}}\nc\bfU{{\bf U}}\nc\cU{{\mathcal U}}
\nc\bfv{{\boldsymbol v}}\nc\bfV{{\bf V}}\nc\cV{{\mathcal V}}
\nc\bfw{{\boldsymbol w}}\nc\bfW{{\bf W}}\nc\cW{{\mathcal W}}\nc\sW{{\mathscr W}}
\nc\bfx{{\boldsymbol x}}\nc\bfX{{\boldsymbol X}}\nc\cX{{\mathcal X}}\nc\sX{{\mathscr X}}
\nc\bfy{{\boldsymbol y}}\nc\bfY{{\boldsymbol Y}}\nc\cY{{\mathcal Y}}\nc\sY{{\mathscr Y}}
\nc\bfz{{\boldsymbol z}}\nc\bfZ{{\boldsymbol Z}}\nc\cZ{{\mathcal Z}}\nc\sZ{{\mathscr Z}}
\title{Extended Cyclic Codes Sandwiched Between Reed--Muller Codes}
\author{%
  \IEEEauthorblockN{Yan Xu, Changjiang Ji, Ran Tao, Sihuang Hu}
  \IEEEauthorblockA{
    \thanks{Y. Xu and C. Ji are with School of Cyber Science and Technology, Shandong University, Qingdao, Shandong, 266237, China. 
    R. Tao and S. Hu are with 
      School of Cyber Science and Technology,
      Shandong University, Qingdao, Shandong, 266237, China, and 
    Key Laboratory of Cryptologic Technology and Information Security, Ministry of Education,
      Shandong University, Qingdao, Shandong, 266237, China. 
      Research partially funded by National Key R\&D Program of China under Grant No. 2021YFA1001000,
      National Natural Science Foundation of China under Grant No. 12001322, 
      Shandong Provincial Natural Science Foundation under Grant No. ZR202010220025,
      and a Taishan scholar program of Shandong Province.
      This work was in part presented at ISIT 2021.
      Email: \{yanxu,jichangjiang\}@mail.sdu.edu.cn; \{rtao,husihuang\}@sdu.edu.cn}
    }
  }

\begin{document}

\maketitle

\begin{abstract}
  The famous Barnes--Wall lattices can be obtained by applying Construction D to a chain of Reed--Muller codes.
  By applying Construction ${\textup{D}}^{\textup{(cyc)}}$ to a chain of extended
  cyclic codes sandwiched between Reed--Muller codes, Hu and Nebe (J. London Math. Soc. (2) 101 (2020) 1068-1089)
  constructed new series of universally strongly perfect lattices sandwiched between
  Barnes--Wall lattices. In this paper, we first extend their construction to generalized Reed--Muller codes,
  and then explicitly determine the minimum vectors of those new sandwiched Reed--Muller codes for some special cases.
\end{abstract}

\section{Introduction}\label{sec_introduction}

Reed--Muller (RM) codes were introduced by Muller~\cite{muller1954application} in 1954,
and shortly after Reed~\cite{reed1954textordfemininea} developed a decoding algorithm. 
RM codes are among the oldest, simplest and perhaps most ubiquitous family of codes.
They are used in many areas related to coding theory, such as electrical engineering and computer science~\cite{MR4190995,MR3666961,massey1992deep,arikan2010survey}. 
Moreover, due to their favourable theoretical and mathematical properties, 
RM codes have also been extensively studied in theoretical computer science~\cite{MR3400278}.
Despite the simplicity of their construction,
many of their important properties are still under investigation~\cite{abbe2020reed}.

The minimum distance of $r$-th order binary RM codes $\cR(r,n)$ is $2^{n-r}$.
The set of minimum vectors (minimum weight codewords) of binary RM codes
was first described by Kasami et al.~\cite{kasami1968new}.
Delsarte, Goethals and MacWilliams~\cite{delsarte1970generalized} characterized
the set of minimum vectors of generalized RM codes. In 1998, Charpin~\cite{charpin1998open}
gave a new proof of the results of~\cite{delsarte1970generalized} for some particular case of
generalized RM codes using Newton identities.
In~\cite{leducq2012new} Leducq proved the same results using arguments from affine geometry.

The famous Barnes--Wall lattices $\BW _{2m}$ of dimension $2^{2m}$ (with $m\in \NN$) form an
important infinite family of even lattices~\cite{BarnesWall,BroueEnguehard}.
One way to construct $\BW_{2m}$ is applying Construction D~\cite[Chapter 8]{conway2013sphere} to a chain of RM codes.
Recently, Hu and Nebe~\cite{hu2020strongly} constructed new series of $2^{2m}$-dimensional
universally strongly perfect lattices sandwiched between Barnes--Wall lattices.
Those lattices can be obtained by applying Construction ${\textup{D}}^{\textup{(cyc)}}$~\cite[Definition 2.5]{hu2020strongly}
to a chain of extended cyclic codes sandwiched between binary RM codes.

In this paper, we first generalize the construction in~\cite{hu2020strongly} to give
new extended cyclic codes sandwiched between generalized RM codes. Then we
employ the tools developed by Augot~\cite{augot1993etude,augot1996description} and Charpin~\cite{charpin1998open} to explicitly
determine the minimum vectors of those new sandwiched RM codes for some special cases.
The paper is organized as follows.
In Section~\ref{sec_pre}, we briefly introduce some basic notations and relevant knowledge of
cyclic codes, Newton identities, affine polynomials and affine invariant codes.
Then we give the construction of extended cyclic codes sandwiched between RM codes
in Section~\ref{sec_sandwiched_codes}. In Section~\ref{sec_properties}, we present the dual code
and the minimum distance of those new sandwiched RM codes.
In Section~\ref{sec_minimum_vectors} we show how to find the minimum vectors of those codes.
Some concluding remarks are given in Section~\ref{sec_conclu}.

\section{Preliminaries}\label{sec_pre}
The following notation will be used throughout this paper.
Let $p$ be a prime, $l$ a positive integer, $q=p^l$, and $\FF_q$ the finite field of order $q$. Let $n$ be a positive integer, $N=q^n-1$, and $\FF_{q^n}$ the finite field of order $q^n$.
\subsection{Group algebra}
Let $\mathcal M$ be the group algebra $\FF_q[\FF_{q^n}^*,\times]$, 
where $\FF_{q^n}^*$ is the multiplicative group of $\FF_{q^n}$. 
An element of $\cM$ is a formal sum
$$x=\sum_{g\in \FF_{q^n}^*}x_g(g),\ x_g\in \FF_q.$$
The scalar multiplication, addition and multiplication in $\cM$ are given as
\begin{align*}
    \lambda\left( \sum_{g\in \FF_{q^n}^*}x_g(g)\right) &=\sum_{g\in \FF_{q^n}^*}\lambda x_g(g),\\
    \sum_{g\in \FF_{q^n}^*}x_g(g)+\sum_{g\in \FF_{q^n}^*}y_g(g)&=\sum_{g\in \FF_{q^n}^*}(x_g+y_g)(g),\\
    \bigg(\sum_{g\in \FF_{q^n}^*}x_g(g)\bigg)\bigg(\sum_{g\in \FF_{q^n}^*}y_g(g)\bigg)&=\sum_{g\in \FF_{q^n}^*}\Big( \sum_{hk=g\atop{h,k\in \FF_{q^n}^*}}x_hy_k \Big)(g).
\end{align*}
We define the following $\FF_q$-linear map from $\mathcal M$ into $\FF_{q^n}$:
$$\rho_s\bigg(\sum_{g\in {\FF_{q^n}^*}}x_g(g)\bigg)=\sum_{g\in {\FF_{q^n}^*}}x_g(g^s),\ s\in \ZZ_{\ge 0}.$$
Note that $\rho_s = \rho_{(s \mod N)}$ as $\FF_{q^n}^*$ is a multiplicative group of order $N$.

Let $\mathcal A$ be the group algebra $\FF_q[\FF_{q^n},+]$. An element of $\cA$ is a formal sum
$$x=\sum_{g\in \FF_{q^n}}x_g(g),\ x_g\in \FF_q.$$
The operations in $\cA$ are given as
\begin{align*}
    \lambda\left( \sum_{g\in \FF_{q^n}}x_g(g)\right) &=\sum_{g\in \FF_{q^n}}\lambda x_g(g),\\
    \sum_{g\in \FF_{q^n}}x_g(g)+\sum_{g\in \FF_{q^n}}y_g(g)&=\sum_{g\in \FF_{q^n}}(x_g+y_g)(g),\\
    \bigg(\sum_{g\in \FF_{q^n}}x_g(g)\bigg)\bigg(\sum_{g\in \FF_{q^n}}y_g(g)\bigg)&=\sum_{g\in \FF_{q^n}}\Big( \sum_{h+k=g\atop{h,k\in \FF_{q^n}}}x_hy_k \Big) (g).
\end{align*}
Note that the multiplication in $\cM$ and $\cA$ are distinct.

\subsection{Cyclic codes}
A linear code $\cC$ is \emph{cyclic} if any cyclic shift of a codeword is also a codeword, i.e., 
whenever $(c_0,...,c_{n-1})\in\cC$ then so is $(c_{n-1},c_0,...,c_{n-2})$.
Recall that $N=q^n-1$.
Set $[0,N-1]=\{0,1, \cdots, N-1\}$. Let $\alpha$ be a primitive element of $\FF_{q^n}$.
We associate with the vector $(c_0,c_1,...,c_{N-1})$ in $\FF_q^{N}$ the element $\sum_{i=0}^{N-1}c_i(\alpha^i)$ in $\cM$.
It is not hard to check that a (primitive) cyclic code of length $N$ over $\FF_q$ is just an ideal of $\mathcal M$. 
%In this case, $G^*=\FF_{q^n}^*$ and $G=\FF_{q^n}$. 
The \emph{defining set $T$} of $\mathcal C$ is the largest subset of the range $[0,N-1]$,
invariant under the multiplication by $q \pmod N$, such that any codeword $x\in \mathcal C$
satisfies $\rho_s(x)=0,\ \forall s\in T.$ The set $T$ is a union of $q$-cyclotomic cosets modulo $N$. 
Any $s\in T$ corresponds to a zero of $\cC$, say $\alpha^s$.
  The set $Z = \{\alpha^s\ | \ s \in T\}$ is the \emph{zero set} of $\cC$.
If $0\notin T$ we denote the \emph{extended code} of $\cC$ as
$$		\widehat{\cC} = \left\{ \bigg(-\sum_{g\in {\FF_{q^n}^*}}x_g\bigg)(0)+ \sum_{g\in {\FF_{q^n}^*}}x_g(g) ~\Bigg|~ \sum_{g\in {\FF_{q^n}^*}}x_g(g) \in \cC\right\}.$$
	By convention, the attached symbol is labelled by $0$, and the defining set of $\widehat{\cC}$ is $T\cup\{0\}$.

Let $x\in \mathcal C$. The \emph{Mattson--Solomon (MS) polynomial}~\cite{mattson1961new} of $x$ is defined by
\begin{equation*}
  M_x(X)=\sum_{s=0}^{N-1}\rho_{N-s}(x)X^s~,
\end{equation*}
whose coefficients are in $\FF_{q^n}$.
The \emph{support} of a codeword $x$ with Hamming weight $w$ is
$$supp(x)=\{g\in \FF_{q^n}^* \mid x_g\neq0\}=\{g_1, \ldots, g_w\}.$$
The \emph{locator polynomial} of $x$ is the polynomial over $\FF_{q^n}$ defined as
$$\sigma_x(X)=\prod_{i=1}^{w}(1-g_iX)=1+\sum_{j=1}^w\sigma_jX^j.$$
The coefficients $\sigma_j, 1\leqslant j\leqslant w$, are the elementary symmetric functions of the \emph{locators} $g_i$, $1\leqslant i\leqslant w$, that is,
$$\sigma_j=(-1)^j\sum_{1\leqslant i_1< i_2<\ldots<i_j\leqslant w}g_{i_1}g_{i_2}\cdots g_{i_j}.$$
For convenience we also write $\Lambda_{N-s}=\rho_{N-s}(x)$ for $0\leqslant s\leqslant N$ and note that $\Lambda_0=\Lambda_N$.

\subsection{Newton identities}
The Newton identities is a tool of great interest for describing a set of codewords, particularly for minimum vectors of cyclic codes~\cite{augot1993etude}\cite{augot1996description}\cite{charpin1998open}.

\begin{theorem}~\cite{augot1996description}\cite[Theorem 3.5]{charpin1998open}
	Let $x\in \mathcal{M}$ be a codeword of weight $w$.
	Let $\Lambda_1, \ldots, \Lambda_N$ be the coefficients of the MS polynomial of $x$ and denote by $\sigma_0, \ldots, \sigma_w$ the coefficients of the locator polynomial of $x$ (note that $\sigma_0 = 1$). Then the $\sigma_i$ and the $\Lambda_i$ are linked by the generalized Newton identities:
	\begin{align}\label{newton_identities}
		\forall\, j\geqslant0,\ \Lambda_{j+w}+\sigma_1\Lambda_{j+w-1}+\ldots+\sigma_w\Lambda_j=0.
	\end{align}
\end{theorem}

%Let the field $\FF_{q^n}$ be the splitting field of $(X^N-1)$. 
We now consider the ring $\FF_{q^n}[\Lambda_0, \ldots, \Lambda_{N-1}, \sigma_1, \ldots, \sigma_w]$ 
and the algebraic system provided from a given cyclic code $\mathcal{C}$ by the Newton identities.
%Let $q = p^l$, where $p$ is the characteristic of the ambient space $\mathcal{M}$. 
Let $\mathcal{C}$ be a cyclic code in $\mathcal{M}$ with the defining set $T$. We define the system $\mathcal{S_C}(w)$, where the $\Lambda_i$ and the $\sigma_i$ are the indeterminates, as follows:
\begin{equation}\label{algebraic_system}
	\mathcal{S_C}(w)=	\left\{
	\begin{array}{l}
		\Lambda_{w+1}+\Lambda_w\sigma_1+\ldots+\Lambda_1\sigma_w=0~,\\
		\Lambda_{w+2}+\Lambda_{w+1}\sigma_1+\ldots+\Lambda_2\sigma_w=0~,\\
		\qquad\vdots\\
		\Lambda_{w+N}+\Lambda_{w+N-1}\sigma_1+\ldots+\Lambda_N\sigma_w=0~,\\
		\forall i\in [0,N-1]~,~~\Lambda_{qi~mod~N}=\Lambda_i^q~,\\
		\forall i\in [0,N-1]~,~~\Lambda_{i+N}=\Lambda_i~,\\
		\forall i\in T~,~~\Lambda_i=0~.
	\end{array} \right.
\end{equation}
The system $\mathcal{S_C}(w)$, consisting of $3N+|T|$ polynomials in $N+w$ variables over the field $\FF_{q^n}$, defines an ideal in the ring
$\FF_{q^n}[\Lambda_0, \ldots, \Lambda_{N-1}, \sigma_1, \ldots, \sigma_w]$.

\begin{corollary}~\cite[Corollary 3.8]{charpin1998open}
	Let $\mathcal{C}$ be a cyclic code in $\mathcal{M}$ whose minimum distance $\delta$ satisfies $\delta\geqslant w$. 
  Then the minimum distance of $\mathcal{C}$ is exactly $w$ if and only if the system $\mathcal{S_C}(w)$ has at least one solution $(\Lambda_0, \ldots, \Lambda_{N-1}, \sigma_1, \ldots, \sigma_w)\in\FF_{q^n}^{N+w}$. For any such solution the codeword $x$ whose MS polynomial is $\sum^{N-1}_{s=0}\Lambda_{N-s}X^s$, is a codeword of $\mathcal{C}$ of weight $w$. The $\sigma_i, 0\leqslant
  i\leqslant w$, are the coefficients of the locator polynomial of $x$. The number of codewords of weight $w$ in $\mathcal{C}$ is equal to the number of solutions of $\mathcal{S_C}(w)$.
\end{corollary}

\subsection{Affine polynomials}
An \emph{affine polynomial} is a polynomial of the form
$$f(X)=c+\gamma_0X+\gamma_1X^q+\ldots+\gamma_kX^{q^k}$$
with coefficients in some extension field of $\FF_q$.
It is well known that the roots of $f(X)$ form an affine space over $\FF_q$.
Since the derivative of $f(X)$ is $\gamma_0$,
the roots of $f(X)$ have multiplicity one if and only if $\gamma_0\neq0$.
So $\gamma_0\neq0$ is necessary if we want to consider $f(X)$ as a locator polynomial.

%We here consider only primitive cyclic code, i.e. $N=q^n-1$, the multiplicative group of $\mathcal{M}$ 
%will be $\FF^*$ ($\FF$ is the field of order $q^n$).
Assume that the affine polynomial $f(X)$ splits in $\FF_{q^n}$ with roots of multiplicity one.
Then $f(X)$ can be identified to be the codeword of $\mathcal{A}$ whose support is
the affine space of their roots and whose symbols are from $\{0, 1\}$.
More precisely, define
\begin{align}\label{affine_codewords}
	x=\lambda \sum_{g\in V_k}(g+h),~\lambda\in\FF_q,
\end{align}
where $V_k$ is a subspace of $\FF_{q^n}$ of dimension $k$ over $\FF_q$.
Such a codeword can be identified, up to scalar multiplication,
with an affine polynomial whose roots are the elements of the affine space $h+V_k$.

The following result presents the form of the locator polynomial of a codeword whose support is coming from some subspace.
\begin{proposition}~\cite[Proposition 3.14]{charpin1998open} \label{THE ROOT}
	Let $\delta=q^{k}-1,\delta<N$, and set
	\[
	\cI_{k}=\{q^{k}-q^{j} \mid j\in [0,k-1]\}.
	\]
	Define the polynomial of degree $\delta$ ($\sigma_{\delta}\ne 0$):
	\[
	\sigma(X)=1+\sum_{i\in \cI_{k}}\sigma_iX^{i}, \sigma_i\in \FF_{q^n}.
	\]
	Denote by $v_i$ the root of $\sigma(X)$ and set $g_i=v_i^{-1}$. Then (i) and (ii) are equivalent.
	\begin{itemize}
		\item[(i)]
		$\sigma(X)$ splits in $\mathbb{F}_{q^n}$ with roots of multiplicity one.
		\item[(ii)]
		$\sigma(X)$ is the locator polynomial of codewords of the form~\eqref{affine_codewords}
		such that $h=0$ and $V_k$ is the set $\{0,g_1,\cdots,g_{\delta}\}$.
	\end{itemize}
\end{proposition}

Let $u\in [0,N]$. We write the $q$-ary expansion of $u$ as
$u=\sum_{i=0}^{n-1}u_iq^i,\ u_{i}\in [0,q-1]$. The $q$-\emph{weight} of $u$ is
$$wt_q(u)=\sum_{i=0}^{n-1}u_i.$$

\begin{theorem}~\cite[Theorem 9]{kasami1968new}\label{PowerSums}
	Let $V_{k}$ be any subspace of $\FF_{q^n}$ of dimension k over $\FF_q$, the field of order $q$.
	Then the power sum
	\[
	\sum_{v\in V_k}v^{i}, \ i\in [1,q^n-1],
	\]
	are zero when the $q$-weight of $i$ is less than $k(q-1)$.
\end{theorem}

\subsection{Affine invariant codes}
%Considering codes of $\mathcal{A}$.
Addition and multiplication in the field $\FF_{q^n}$ involve natural transformations on elements of $\cA$ including the following affine permutations
$$\sigma_{u,v}:\sum_{g\in\FF_{q^n}}x_g(g)\longmapsto\sum_{g\in\FF_{q^n}}x_g(ug+v),~~u\in\FF_{q^n}^*,~~v\in\FF_{q^n}.$$
The permutations $\sigma_{u,0}$ consist of shifting symbols unless the symbol labelled by ``$0$".
It is exactly the shift on codewords punctured in the position ``$0$". On the other hand we have
$$\sigma_{1,v}(\sum_{g\in\FF_{q^n}}x_g(g))=(v)(\sum_{g\in\FF_{q^n}}x_g(g)).$$
Hence we see that a code $\mathcal{C}\subseteq \cA$, which is invariant under the permutations $\sigma_{u,v}$ $(u\in\FF_{q^n}^*,v\in\FF_{q^n})$, is an ideal of $\mathcal{A}$. Such a code is called an \emph{affine-invariant code}.

Kasami et al.~\cite{kasami1966weight,kasami1967some} showed that an extended cyclic code is affine-invariant if and only if its defining set satisfies a certain combinatorial property as follows.
Let $S=[0, q^{n}-1]$. We denote by $\preceq$ the partial order relation on $S$ defined as follows:
\[
\forall s,t\in S : s\preceq t  \Leftrightarrow s_i\leqslant t_i, i\in [0,n-1],
\]
where $\sum_{i=0}^{n-1} s_iq^i$ is the $q\text{-}ary$ expansion of $s$
and $\sum_{i=0}^{n-1} t_iq^i$ is the $q\text{-}ary$ expansion of $t$.

\begin{lemma}\cite[Theorem 2.14]{kasami1966weight,kasami1967some,charpin1998open}\label{affine_invariant}
	Let us define the map
	\[
	\Delta : T\subset S \mapsto \Delta(T)=\bigcup_{t\in T}\{s\in S,\ s\preceq t\}.
	\]
	Let $\cC$ be an extended cyclic code, with the defining set $T$.
	Then $\cC$ is affine-invariant if and only if $\Delta(T)=T$.
\end{lemma}

\section{Definitions of sandwiched Reed--Muller codes}\label{sec_sandwiched_codes}
From now on, we set $n=2m$ and $\FF_{q^n}=\FF_q(\alpha)$ where $\alpha$ is a primitive element of $\FF_{q^{n}}$. 
The following notations will be used.

\begin{notation}\label{wtq}
\begin{itemize}
	\item[(a)]
	Any integer $u$, $0\leqslant u \leqslant q^{n}-1$, has a unique expression as
	$u= \sum _{i=0}^{{n}-1} u_i q^i $ with $u_i \in \{0,1,\cdots,q-1\}$.
%	Then, the \emph{$q$-weight} of $u$ is
%	\begin{align*}
%		\wt_q(u) := \sum_{i=0}^{{n}-1}u_i.
%	\end{align*}
	We define
	\begin{align*}
		O(u)&:= \sum_{0\leqslant i\leqslant {n}-1\atop i\textup{~is~odd}}u_i
		\textup{ and }\\
		E(u)&:= \sum_{0\leqslant i\leqslant {n}-1\atop i\textup{~is~even}}u_i .
	\end{align*}
    Note that $q^n-1-u=\sum_{i=0}^{{n}-1}(q-1-u_i)q^i$. Hence
    \begin{align}\label{wt_qn-1}
      \begin{split}
      \wt_q(q^n-1-u)=n(q-1)-\wt_q(u),\\
      O(q^n-1-u)-E(q^n-1-u)=E(u)-O(u).
      \end{split}
    \end{align}
	\item[(b)]
	For $-1\leqslant r<{n}(q-1)$, put
	$$ Z_r  :=  \{ \alpha^u \mid 0<u\leqslant q^{n}-1, \wt_q(u)\leqslant {n}(q-1)-r-1 \} . $$
	\item[(c)]
	For $0\leqslant r\leqslant {n}(q-1)$, let
	$$ \Theta ^{(r)}
	:=\{ \alpha^u \mid 0\leqslant u \leqslant q^{n}-1 , \wt_q(u) = {n}(q-1)-r \}.  $$

	\item[(d)]
	For $0\leqslant k \leqslant m(q-1)$, we put
	$$\Theta_k:=\{ \alpha^u \mid 0\leqslant u\leqslant q^{n}-1, \ |O(u) - E(u)| =k \}. $$

	\item[(e)]
    Note that $\wt_q(u)$ and $|O(u) - E(u)|$ have the same parity.
	For $0\leqslant r\leqslant {n}(q-1)$, let
  \begin{align*}
    M_r&:=\\
    &\begin{cases}
    M_+ := \{ 0\leqslant k \leqslant m(q-1)~|~k \textup{~is even}\} & \textup{if } r \textup{\ is even} \\
    M_- := \{ 0\leqslant k \leqslant m(q-1)~|~k \textup{~is odd} \} & \textup{if } r \textup{\ is odd.}\\
    \end{cases}
  \end{align*}
	For $0\leqslant r\leqslant {n}(q-1)$ and  $k\in M_r$,  we  define
	\begin{align*}
		\Theta ^{(r)}_{k} :&= \{ \alpha^u \mid 0\leqslant u \leqslant q^{n}-1, \wt_q(u) = {n}(q-1)-r,\\
		&\qquad|O(u)-E(u)| =k \}  \\&= \Theta ^{(r)} \cap \Theta _k .
	\end{align*}
 	Obviously $\Theta ^{(r)} \cap \Theta _k = \emptyset $ if $k\not\in M_r $. 
 	It is ready to see that $\Theta ^{(r)} \cap \Theta _k$ is a union of $q$-cyclotomic cosets modulo $N$ as $n=2m$ is even.
\end{itemize}
\end{notation}

In order to compute the dimension of our new codes, now we use a counting argument to show the size of the set $\Theta ^{(r)}_{k}$.

\begin{lemma}\label{theta^r_k}
    For $0\leqslant r\leqslant {n}(q-1)$ and $k\in M_r$, we have
    \begin{align*}
      \big|\Theta ^{(r)}_{k}\big| = 
      \begin{cases}
          2\left(\sum_{i=0}^{m}(-1)^i{m\choose i}{\frac{n(q-1)-r-k}{2}-iq+m-1\choose \frac{n(q-1)-r-k}{2}-iq}\right) \left(\sum_{i=0}^{m}(-1)^i{m\choose i}{\frac{n(q-1)-r+k}{2}-iq+m-1\choose \frac{n(q-1)-r+k}{2}-iq}\right) & k \neq 0 \\
          \left(\sum_{i=0}^{m}(-1)^i{m\choose i}{\frac{n(q-1)-r}{2}-iq+m-1\choose \frac{n(q-1)-r}{2}-iq}\right)^2 & k=0.
      \end{cases}
    \end{align*}
\end{lemma}

\if{false}
%\begin{proof}
    Note that $|\Theta ^{(r)}_{k}| = |\{u: 0\leqslant u \leqslant q^{n}-1, \wt_q(u) = {n}(q-1)-r,|O(u)-E(u)| =k \}|$.
    Write $u= \sum _{i=0}^{{n}-1} u_i q^i $ with $u_i \in \{0,1,\cdots,q-1\}$.
    Firstly, we assume that $O(u)-E(u)=k$. Combining this with $O(u)+E(u)=n(q-1)-r$ we have
    \begin{align*}
    E(u)&=\frac{n(q-1)-r-k}{2}, \\
    O(u)&=\frac{n(q-1)-r+k}{2}.
    \end{align*}
  There are $m$ $u_i$'s when $0\leqslant i\leqslant n-1$ and $i$ is even, $u_0,u_2,\cdots,u_{n-2}$, as well as $m$ $u_i$'s when $0\leqslant i\leqslant n-1$ and $i$ is odd, $u_1,u_3,\cdots,u_{n-1}$. An inclusion-exclusion argument shows that the inner sum $\sum_{i=0}^{m}(-1)^i{m\choose i}{\frac{n(q-1)-r-k}{2}-iq+m-1\choose \frac{n(q-1)-r-k}{2}-iq}$ is the number of ways of picking $\frac{n(q-1)-r-k}{2}$ objects from a set of $m$ objects, when no object can be chosen more than $q-1$ times. Similarly, $\sum_{i=0}^{m}(-1)^i{m\choose i}{\frac{n(q-1)-r+k}{2}-iq+m-1\choose \frac{n(q-1)-r+k}{2}-iq}$ is the number of ways of picking $\frac{n(q-1)-r+k}{2}$ objects from a set of $m$ objects, when no object can be chosen more than $q-1$ times.
  
  Next, the number of $u$ satisfying these conditions is 
  $$\left(\sum_{i=0}^{m}(-1)^i{m\choose i}{\frac{n(q-1)-r-k}{2}-iq+m-1\choose \frac{n(q-1)-r-k}{2}-iq}\right) \left(\sum_{i=0}^{m}(-1)^i{m\choose i}{\frac{n(q-1)-r+k}{2}-iq+m-1\choose \frac{n(q-1)-r+k}{2}-iq}\right).$$
  By the same way, if $E(u)-O(u)=k\neq0$, we get the same number. Hence, 
  $$\big|\Theta^{(r)}_{k}\big| =2\left(\sum_{i=0}^{m}(-1)^i{m\choose i}{\frac{n(q-1)-r-k}{2}-iq+m-1\choose \frac{n(q-1)-r-k}{2}-iq}\right) \left(\sum_{i=0}^{m}(-1)^i{m\choose i}{\frac{n(q-1)-r+k}{2}-iq+m-1\choose \frac{n(q-1)-r+k}{2}-iq}\right).$$
  When $k=0$, $O(u)-E(u)=E(u)-O(u)$, so 
  $$\big|\Theta^{(r)}_{0}\big| =\left(\sum_{i=0}^{m}(-1)^i{m\choose i}{\frac{n(q-1)-r}{2}-iq+m-1\choose \frac{n(q-1)-r}{2}-iq}\right)^2.$$ 
%\end{proof}
\fi

\begin{proof}
Firstly, we set 
\begin{align*}
  \Delta^{(r)}_k&=\{u \mid 0\leqslant u\leqslant q^n-1,\wt_q(u)=n(q-1)-r,|O(u)-E(u)|=k\},\\
  S_1&=\{u\in\Delta^{(r)}_k\mid O(u)-E(u)=k\},\\
  S_2&=\{u\in\Delta^{(r)}_k\mid E(u)-O(u)=k\}.
\end{align*}
Then we have $|S_1|=|S_2|$ and
$$|\Theta^{(r)}_{k}|=|\Delta^{(r)}_k|=
\begin{cases}
2|S_1| & \mbox{~if~} k\neq0\\
|S_1| & \mbox{~if~} k=0.
\end{cases}
$$
Now we compute $|S_1|$. Suppose that $u\in S_1$. Then it is ready to check that
\begin{align*}
    E(u)&=\frac{n(q-1)-r-k}{2}, \\
    O(u)&=\frac{n(q-1)-r+k}{2}.
\end{align*}
Write $u= \sum _{i=0}^{{n}-1} u_i q^i $ with $u_i \in \{0,1,\cdots,q-1\}$, and
\begin{align*}
 \widetilde{E}&=\left\{(u_0,u_2,\cdots,u_{n-2})\mid u_i\in[0,q-1],\sum_{0\leqslant i\leqslant {n}-1\atop i\textup{~is~even}}u_i=\frac{n(q-1)-r-k}{2} \right\},\\
\widetilde{O}&=\left\{(u_1,u_3,\cdots,u_{n-1})\mid u_i\in[0,q-1],\sum_{0\leqslant i\leqslant {n}-1\atop i\textup{~is~odd}}u_i=\frac{n(q-1)-r+k}{2} \right\}.   
\end{align*}
It is not hard to check that $|\widetilde{E}|$ is the number of ways of picking $\frac{n(q-1)-r-k}{2}$ objects from a set of $m$ objects, while no object can be chosen more than $q-1$ times. By the principle of inclusion and exclusion, we have
$$|\widetilde{E}|=\sum_{i=0}^{m}(-1)^i{m\choose i}{\frac{n(q-1)-r-k}{2}-iq+m-1\choose \frac{n(q-1)-r-k}{2}-iq}.$$ 
Similarly, 
$$|\widetilde{O}|=\sum_{i=0}^{m}(-1)^i{m\choose i}{\frac{n(q-1)-r+k}{2}-iq+m-1\choose \frac{n(q-1)-r+k}{2}-iq}.$$
Hence, 
$$|S_1|=|\widetilde{E}||\widetilde{O}|=\left(\sum_{i=0}^{m}(-1)^i{m\choose i}{\frac{n(q-1)-r-k}{2}-iq+m-1\choose \frac{n(q-1)-r-k}{2}-iq}\right) \left(\sum_{i=0}^{m}(-1)^i{m\choose i}{\frac{n(q-1)-r+k}{2}-iq+m-1\choose \frac{n(q-1)-r+k}{2}-iq}\right).$$
Then the claim follows directly.
\end{proof}

\begin{definition}\cite[Theorem 1]{kasami1968new} \cite[Chapter 13, Theorem 11]{macwilliams1977theory} \cite[Definition 5.14, Theorem 5.18]{assmus1998polynomial}
  For $-1\leqslant r<{n}(q-1)$, define $\cR_q(r,{n})^*$ to be the length $q^{n}-1$ cyclic code with the zero set $Z_r$ where $Z_r$ is as in Notation~\ref{wtq} (b).
  The extended code of $\cR_q(r,{n})^*$ is the \emph{$r$th-order generalized Reed--Muller} code $\cR_q(r,{n})$.
  Note that $\cR_q (-1,{n})=\{0\}$, and $\cR_q ({n},{n})=\FF_q^{q^n}$ which is not an extended cyclic code.
\end{definition}

Kasami et al.~\cite{kasami1968new} first determined the minimum distance of $\cR_q(r,{n})^*$,
then later Delsarte et al.~\cite{delsarte1970generalized} gave a complete description of the minimum vectors. For a code $\cC$ we will use $\dist(\cC)$ to denote its minimum distance.

\begin{theorem}\cite[Theorem 5]{kasami1968new}\label{RM_min_dist}%\cite[Theorem 2.6.2]{delsarte1970generalized}\cite[Corollary 5.26]{assmus1998polynomial}
  For $0\leqslant r<{n}(q-1)$, write $r=\rho(q-1)+s$, $0\leqslant \rho <n, 0\leqslant s<q-1$. We have
  \begin{align*}
    \dist(\cR_q(r,{n})^*)&=(q-s)q^{{n}-\rho-1}-1,\\
    \dist(\cR_q(r,{n})) & =(q-s)q^{{n}-\rho-1}.
  \end{align*}
\end{theorem}

 \begin{theorem}\cite[Theorem 2.6.3]{delsarte1970generalized}\cite[Theorem 3.18]{charpin1998open}\label{RM_min_vectors}
  For $0\leqslant r<{n}(q-1)$, write $r=\rho(q-1)+s$, $0\leqslant \rho <n, 0\leqslant s<q-1$. Assume that $s=0$.
  Then the minimum vectors of the punctured code $\cR_q(r,{n})^*$ are the codewords of weight $\delta=q^{{n}-\rho}-1$, whose locators are the nonzero elements of some subspace $V_{{n}-\rho}$ of $\FF_{q^{n}}$ of dimension ${n}-\rho$ over $\FF_q$. These are in the algebra $\mathcal{M}$ precisely the codewords
 $$x=\sum_{g\in V_{{n}-\rho}^{*}}\lambda(g)~,~\lambda\in \FF_q^*,$$
 where $V_{{n}-\rho}^*=V_{{n}-\rho}\backslash\{0\}$.
 And the locator polynomial of $x$ is 
 $$\sigma_x(X)=1-\sum_{s\in \mathcal{I}_{n-\rho}}\frac{\Lambda_{\delta+s}}{\Lambda_{\delta}}X^s,$$
where $\cI_{n-\rho}=\{q^{n-\rho}-q^j \mid j\in[0,n-\rho-1]\}.$

   Similarly, the minimum vectors of the extended code $\cR_q(r,{n})$ are the codewords of weight $q^{{n}-\rho}$ whose locators are the elements of some affine subspace $h+V_{{n}-\rho},~h\in\FF_{q^{n}}$, of $\FF_{q^{n}}$ of dimension ${n}-\rho$. These are in the algebra $\mathcal{A}$ precisely the codewords
  $$x=\lambda \sum_{g\in V_{{n}-\rho}}(g+h),~\lambda\in\FF_q^*.$$
\end{theorem}

Now we define new classes of (extended) cyclic codes sandwiched between generalized RM codes.
For the binary case this was first introduced by Hu and Nebe~\cite{hu2020strongly}. 
\begin{definition}
  Let $0\leqslant r<{n}(q-1) $ and $I\subset M_{r}$ be given. Put $\overline{I}=M_r \setminus I$ and
  $$Z_{r,I} := Z_{r-1} \bigg\backslash\bigg(\bigcup _{k\in I} \Theta ^{(r)}_{k} \bigg)
  =Z_{r} \bigcup \bigg(\bigcup_{k\in \overline{I}} \Theta ^{(r)}_{k}\bigg).$$
  Note that $Z_r \subseteq Z_{r,I} \subseteq Z_{r-1} $.
  Then, let
  $\cC_q(r,I,{n})^* \subseteq \FF _q^{q^{n}-1} $ be the cyclic code with the zero set $Z_{r,I}$
  and $\cC_q(r,I,{n}) \subseteq \FF _q^{q^{n}} $ the extended code of $\cC_q(r,I,{n})^*$.
  Also we define
  $$
  \cC_q({n(q-1)},I,{n}) =
  \begin{cases}
    \cR_q({n(q-1)}-1,n) & \textup{if } 0\notin I,\\
    \cR_q({n(q-1)},{n})=\FF_q^{q^{n}} &\textup{otherwise.}
  \end{cases}
  $$
\end{definition}

\begin{remark}\label{CIdetails}
  \begin{itemize}
    \item[(a)] $\cR_q (r-1,{n}) \subseteq \cC_q(r,I,{n}) \subseteq \cR_q (r,{n}) $.
    \item[(b)] $\cR_q (r-1,{n}) = \cC_q(r,\emptyset,{n}) $.
    \item[(c)] $\cR_q (r,{n}) = \cC_q(r,M_r,{n}) $.
    \item[(d)] If $I\subseteq J \subseteq M_r$ then
    $\cC_q(r,I,{n}) \subseteq \cC_q(r,J,{n}) $.
  \end{itemize}
\end{remark}

Because of $(a)$ of Remark~\ref{CIdetails} we call those $\cC_q(r,I,{n})$ \emph{sandwiched Reed--Muller codes}.
For a code $C$ we will use $\dim(C)$ to denote its dimension. 

\begin{lemma}\cite[Theorem 5.5]{assmus1998polynomial}\label{dim(R)}
    For any $r$ such that $0\leqslant r\leqslant n(q-1)$,
    $$\dim(\cR_q(r,n))=\sum_{i=0}^n(-1)^i{n\choose i}{r-iq+n\choose r-iq}.$$
\end{lemma}

Here we give the dimension of our new codes by counting their zeros.

\begin{theorem}
    For any $r$ such that $0\leqslant r< n(q-1)$ and $I\subset M_{r}$, 
    \begin{align*}
        \dim(\cC_q(r,I,n))&=\sum_{i=0}^n(-1)^i{n\choose i}{n+r-iq\choose r-iq}\\
        &\ \ -2\sum_{k\in\overline{I}\atop k\neq0}\left(\sum_{i=0}^{m}(-1)^i{m\choose i}{\frac{n(q-1)-r-k}{2}-iq+m-1\choose \frac{n(q-1)-r-k}{2}-iq}\right) \left(\sum_{i=0}^{m}(-1)^i{m\choose i}{\frac{n(q-1)-r+k}{2}-iq+m-1\choose \frac{n(q-1)-r+k}{2}-iq}\right)\\
        &\ \ -\mathbf{1}_{\overline{I}}(0)\left(\sum_{i=0}^{m}(-1)^i{m\choose i}{\frac{n(q-1)-r}{2}-iq+m-1\choose \frac{n(q-1)-r}{2}-iq}\right)^2,
    \end{align*}
    where $\mathbf{1}_{\overline{I}}(0)=1$ if $0\in\overline{I}$ and $\mathbf{1}_{\overline{I}}(0)=0$ otherwise.
\end{theorem}
\begin{proof}
    For a cyclic code $\cC$ of length $l$ we have $\dim(\cC)+|Z(\cC)|=l$ where $Z(\cC)$ is the zero set of $\cC$. 
    %For $\cR_q(r,n)^*$ or $\cC_q(r,I,n)^*$, the number of its zeros plus its dimension is equal to $N$.
    So
    $$\dim(\cC_q(r,I,n)^*)+|Z_r|+\sum_{k\in\overline{I}}\big|\Theta_k^{(r)}\big|=\dim(\cR_q(r,n)^*)+|Z_r|=N.$$
    %On the other hand, we have 
    %$$\dim(\cR^*_q(r,n))+|Z_r|=N.$$
    Then
    $$\dim(\cC_q(r,I,n))=\dim(\cC_q(r,I,n)^*)=\dim(\cR_q(r,n))-\sum_{k\in\overline{I}}\big|\Theta_k^{(r)}\big|,$$
    Now the result follows from Lemma~\ref{theta^r_k} and Lemma~\ref{dim(R)}.
\end{proof}

In the following we give an example of our new code.
\begin{example}
Let $q=3,m=2,n=4,r=5$, $I=\{1\}\subseteq M_5=\{1,3\}$ and $\overline{I}=\{3\}$. 
Let $\alpha$ be a primitive element of $\FF_{3^4}$. It's not hard to compute that $\Theta ^{(5)}_{3}=\{\alpha^{11},\alpha^{33},\alpha^{19},\alpha^{57}\}$.
The zeros of $\cR_3(4,4)^*$, $\cR_3(5,4)^*$ and $\cC_3(5,I,4)^*$ are 
\begin{align*}
    Z_4&=\{\alpha^u\ |\ 0<u\leqslant 80, \wt_3(u)\leqslant 3\}\\
       &=\{\alpha^1,\alpha^3,\alpha^9,\alpha^{27},\\
       &\qquad\alpha^2,\alpha^6,\alpha^{18}, \alpha^{54},\alpha^4,\alpha^{12},\alpha^{36},\alpha^{28},\alpha^{10},\alpha^{30},\\
       &\qquad\alpha^5,\alpha^{15},\alpha^{45},\alpha^{55},\alpha^7,\alpha^{21},\alpha^{63},\alpha^{29},\alpha^{11},\alpha^{33},\alpha^{19},\alpha^{57},\alpha^{13},\alpha^{39},\alpha^{37},\alpha^{31}\},\\
    Z_5&=\{\alpha^u\ |\ 0<u\leqslant 80, \wt_3(u)\leqslant 2\}\\
       &=\{\alpha^1,\alpha^3,\alpha^9,\alpha^{27},\\
       &\qquad\alpha^2,\alpha^6,\alpha^{18}, \alpha^{54},\alpha^4,\alpha^{12},\alpha^{36},\alpha^{28},\alpha^{10},\alpha^{30}\},\\
    Z_5\cup\Theta ^{(5)}_{3} &=\{\alpha^u\ |\ 0<u\leqslant 80, \wt_3(u)\leqslant 2 \mbox{ and } \wt_3(u)=3, |O(u)-E(u)|=3\}\\
    &=\{\alpha^1,\alpha^3,\alpha^9,\alpha^{27},\\
       &\qquad\alpha^2,\alpha^6,\alpha^{18}, \alpha^{54},\alpha^4,\alpha^{12},\alpha^{36},\alpha^{28},\alpha^{10},\alpha^{30},\\
       &\qquad \alpha^{11},\alpha^{33},\alpha^{19},\alpha^{57}\},
\end{align*}
respectively. Hence 
$$ \cR_3(4,4)^* \subset \cC_3(5,I,4)^* \subset \cR_3(5,4)^*.$$
By Lemma~\ref{dim(R)}, we obtain
$$\dim(\cR_3(4,4))=\sum_{i=0}^4(-1)^i{4\choose i}{8-3i\choose 4-3i}=50,$$
$$\dim(\cR_3(5,4))=\sum_{i=0}^4(-1)^i{4\choose i}{9-3i\choose 5-3i}=66.$$
By Lemma~\ref{theta^r_k} we have
$$|\Theta ^{(5)}_{3}|=2\left(\sum_{i=0}^{2}(-1)^i{2\choose i}{0-3i+1\choose 0-3i}\right) \left(\sum_{j=0}^{2}(-1)^j{2\choose j}{3-3j+1\choose 3-3j}\right)=4.$$
Hence $$\dim(\cC_3(5,I,4))=\dim(\cR_3(5,4))-|\Theta ^{(5)}_{3}|=62.$$
\end{example}

\begin{theorem}\label{affine_invariant_extended}
  For any $0\leqslant r<{n}(q-1)$ and $I\subset M_{r}$, 
  the code ${\cC_q}(r,I,{n})$ is affine-invariant.
\end{theorem}
\begin{proof}
  Let $Z_{r,I}$ be the zero set of ${\cC_q}(r,I,{n})^*$, and $T_{r,I}$ the defining set of $C_q(r,I,n)$, where 
  $$T_{r,I}=\{i\mid 1\leqslant i\leqslant q^n-1, \alpha^i\in Z_{r,I}\}\cup\{0\}.$$ 
  For any $t\in T_{r,I}$ and $s\in [0,q^n-1]$, if $s\preceq t$ then we have $s=t$ or $\wt_q(s)<\wt_q(t)$, which shows that $s\in T_{r,I}$ according to the definition of $Z_{r,I}$. Hence $\Delta(T_{r,I})=T_{r,I}$. Now by Lemma~\ref{affine_invariant}, the code $C_q(r,I,n)$ is affine-invariant.
\end{proof}

\section{Some properties of sandwiched Reed--Muller codes}\label{sec_properties}

Now we investigate the dual codes and the minimum distances of our new sandwiched Reed--Muller codes.

\subsection{Dual codes}
The dual codes of generalized RM codes first appeared in unpublished notes of S. Lin,
then was given in~\cite[Theorem 2.2.1]{delsarte1970generalized},
$$\cR_q (r,{n})^\perp=\cR_q (n(q-1)-r-1,{n}),\ 0\leqslant r< n(q-1).$$
Here we prove the following result on the dual of sandwiched RM codes.

\begin{theorem}
For $0\leqslant r< n(q-1)$, $I\subseteq M_r$ and $\overline{I}:=M_r\setminus I$, we have 
$$\cC_q(r,I,n)^\perp=\cC_q(n(q-1)-r,\overline{I},n).$$
\end{theorem}

\begin{proof}
	Recall that the zero set of $\cC_q(r,I,n)^*$ is $Z_{r,I}$.
  It is well-known that (c.f. \cite[Definition 2.7]{charpin1998open}) the dual code of $\cC_q(r,I,n)$ is also an extended cyclic code.
  More precisely, $\cC_q(r,I,n)^\perp$ is the extended code of the cyclic code whose zeros are 
  $$ \big\{\beta^{-1}\ |\ \beta\in \FF_{q^n}^*\setminus (Z_{r,I}\cup\{\alpha^0\})\big\},$$
  and by Equation~\eqref{wt_qn-1}
  %Showing in detail, 
  %$$\FF_{q^n}^*\setminus (Z_{r,I}\cup\{\alpha^0\})=\{\alpha^u\ |\ 0<u\leqslant q^n-1,\ n(q-1)-r+1\leqslant\wt_q(u)\leqslant n(q-1)-1\}
  %\cup \{\alpha^u\ |\ 0<u\leqslant q^n-1,\ \wt_q(u)= n(q-1)-r,\ |O-E|=k,\ k\in I\},$$
  it is not hard to check that this zero set is equal to $Z_{n(q-1)-r,\overline{I}}$.
  %=\{\alpha^u\ |\ 0<u\leqslant q^n-1,\ 0<\wt_q(u)\leqslant r-1\}
  %\cup \{\alpha^u\ |\ 0<u\leqslant q^n-1,\ \wt_q(u)= r,\ |O-E|=k,\ k\in I\}.$$
  This concludes our proof.
\end{proof}

\subsection{Minimum distances}
Now we show the following lower bounds on the minimum distances of sandwiched RM codes.

\begin{theorem}\label{thm_distance}
  For $0\leqslant r \leqslant {n}(q-1)-2$, write $r=\rho(q-1)+s$, where $0\leqslant \rho\leqslant n-1, 0\leqslant s<q-1$.
  Suppose $I\subseteq M_{r+1}$.  
  %where $r=\rho(q-1)+s$, $0\leqslant \rho\leqslant n-1, 0\leqslant s<q-1$, we have \footnote{We use $\{ q-1-s\pm1\}$ to denote  $\{ q-1-s+1, q-1-s-1\}$ for convenience.}
  \begin{enumerate}
    \item If $\rho=n-1$ and $\{q-1-s-1\} \cap I = \emptyset$, then $\dist(\cC_q(r+1,I,{n}))=q-s$.
    \item For $0\leqslant \rho\leqslant n-2$, we have\footnote{We use $\{ q-1-s\pm1\}$ to denote  $\{ q-1-s+1, q-1-s-1\}$ for convenience.}
  $$
\dist(\cC_q(r+1,I,{n}))
\begin{cases}
  =(q-s)q^{{n}-\rho-1} & \textup{if~} \rho \textup{\ is\ odd\ and\ } \{q-1-s\pm1\} \cap I = \emptyset, \textup{\ or}\\
  &\quad \textup{if~} \rho \textup{\ is\ even\ and\ } \{|s\pm1|\} \cap I = \emptyset\\
	\geqslant (q^2-qs-1)q^{{n}-\rho-2} & \textup{if } \rho \textup{\ is\ odd\ and\ } \{q-1-s\pm1\} \cap I= \{q-1-s+1\}, \textup{\ or}\\
&\quad \textup{if~} \rho \textup{\ is\ even,\ } s\neq 0\textup{ and } \{|s\pm1|\}\cap I= \{|s-1|\}.
\end{cases}
$$
  \end{enumerate}
\end{theorem}
\begin{proof}
  By (a) of Remark~\ref{CIdetails} and Theorem~\ref{RM_min_dist},
	$$ \dist(\cR_q (r+1,{n}))\leqslant\dist(\cC_q(r+1,I,{n}))\leqslant\dist(\cR_q (r,{n}))=(q-s)q^{{n}-\rho-1}. $$
  We first review the minimum distance of $\cR_q (r,{n})$. 
  Consider the set
	$$Z=\{\alpha^u\ |\ 0<u<(q-s)q^{{n}-\rho-1}-1\}.$$
  Note that 
  \begin{align*}
        (q-s)q^{{n}-\rho-1}-1=
  \begin{cases}
       q-s-1 & \textup{if } \rho=n-1 \\
       (q-s-1)q^{n-\rho-1}+\sum_{i=0}^{n-\rho-2}(q-1)q^i & \textup{if } 0\leqslant\rho\leqslant n-2
  \end{cases}
  \end{align*}
  which implies that for $0<u<(q-s)q^{{n}-\rho-1}-1$, we have $0<\wt_q(u)\leqslant {n}(q-1)-r-1$. 
  So $Z$ is a subset of $Z_r$, the zero set of $\cR_q (r,{n})^*$. Then the BCH bound gives 
  $$\dist(\cR_q (r,{n})^*)\geqslant (q-s)q^{{n}-\rho-1}-1.$$ 
  As $\cR_q (r,{n})$ is affine invariant, we have by translation invariance (c.f.~\cite[Corollary 5.26]{assmus1998polynomial})
  $$\dist(\cR_q (r,{n}))\geqslant (q-s)q^{{n}-\rho-1}.$$ 

  Recall that the zero set of $\cC_q(r+1,I,{n})^*$ is  
  $$Z_{r+1,I} = Z_{r} \bigg\backslash\bigg(\bigcup _{k\in I} \Theta ^{(r+1)}_{k} \bigg).$$

  Suppose that $0<u<(q-s)q^{{n}-\rho-1}-1$ and $\wt_q(u)={n}(q-1)-r-1$.
  It is not hard to verify that there exists some $j, 0\leqslant j\leqslant n-\rho-1$ such that
  \begin{align*}
  u=(q-s)q^{{n}-\rho-1}-1-q^j=
  \begin{cases}
  q-s-1-1 & \textup{if } \rho=n-1 \\
  (q-s-1)q^{n-\rho-1}-q^j+\sum_{i=0}^{n-\rho-2}(q-1)q^i  & \textup{if } 0\leqslant\rho\leqslant n-2.
  \end{cases}
  \end{align*}
  Case (i): If $\rho$ is odd, then 
  $$E(u)-O(u)=
  \begin{cases}
      q-s-1-1 & \textup{if } \rho=n-1\\
      q-s-1\pm1 & \textup{if } \rho\leqslant n-3.
  \end{cases}
  $$
  So if $\rho=n-1$ and $\{q-1-s-1\} \cap I = \emptyset$ then $Z\subseteq Z_{r+1,I}$.
  Also for $0\leqslant \rho\leqslant n-3$, if $\{q-s-1\pm1\}\cap I=\emptyset$, then $Z\subseteq Z_{r+1,I}$. 
  Using a similar argument as above, we can conclude that $\dist(\cC_q(r+1,I,{n}))=\dist(\cR_q (r,{n}))$. 
  On the other hand, we can verify that the minimal $u$ such that
  $E(u)-O(u)=q-s-1+1$ is 
  $$u=\sum_{i=0}^{n-\rho-3}(q-1)q^i+(q-2)q^{{n}-\rho-2}+(q-s-1)q^{{n}-\rho-1}=(q^2-qs-1)q^{{n}-\rho-2}-1,$$
  where $0\leqslant \rho\leqslant n-3$. Therefore, if $\{q-s-1\pm1\}\cap I=\{q-s-1+1\}$ then
	\begin{align*}
		\dist(\cC_q(r+1,I,{n})^*)&\geqslant(q^2-qs-1)q^{{n}-\rho-2}-1.
	\end{align*}
  As $\cC_q (r+1,I,{n})$ is also affine invariant, we have by translation invariance
  $$\dist(\cC_q(r+1,I,{n}))\geqslant(q^2-qs-1)q^{{n}-\rho-2}.$$

  Case (ii): If $\rho$ is even, then $\rho\leqslant n-2$ and $E(u)-O(u)=s\pm1$. Similarly, if $\{|s\pm1|\}\cap I=\emptyset$, we have $Z\subseteq Z_{r+1,I}$
  and $\dist(\cC_q(r+1,I,{n}))=\dist(\cR_q (r,{n}))$.
  Note that the minimal $u$ such that $E(u)-O(u)=s-1$ is also $(q^2-qs-1)q^{n-\rho-2}-1.$
	Therefore, if $s\neq 0$ and $\{|s\pm1|\}\cap I=\{|s-1|\}$, then similarly as above we have
	\begin{align*}
	\dist(\cC_q(r+1,I,n)^*)&\geqslant(q^2-qs-1)q^{n-\rho-2}-1,\\
	\dist(\cC_q(r+1,I,n))&\geqslant(q^2-qs-1)q^{n-\rho-2}.
	\end{align*}
\end{proof}

\begin{example}
    Let $q=3,m=2$ and $n=4$. For $0\leqslant r\leqslant 6$ and $I\subseteq M_{r+1}$, we use $[N,K,D]$ to denote the length, dimension and minimum distance of the code $\mathcal{C}_3(r+1,I,4)$. 
    With the aid of Magma program, we list the parameters of these codes in Tables I-II. 
     \begin{table}[!h]
    	\caption{The parameters of $\mathcal{C}_3(r+1,I,4)$ with even $r$}
    	\centering
    	\begin{tabular}{c|c|c|c|c}
    		\hline
    		\diagbox{$I$}{$[N,K,D]$}{$r$} & 0 & 2 & 4 & 6\\
    		\hline
    		\{1\} & [81,5,54] & [81,27,18] & [81,62,6] & [81,80,2]\\
    		\hline
    		\{3\} & [81,1,81] & [81,19,27] & [81,54,9] & [81,76,3]\\
    		\hline
    		\{1,3\} & [81,5,54] & [81,31,18] & [81,66,6] & [81,80,2]\\
    		\hline
    	\end{tabular}
    \end{table}

\begin{table}[!h]
	\caption{The parameters of $\mathcal{C}_3(r+1,I,4)$ with odd $r$}
	\centering
	\begin{tabular}{c|c|c|c}
		\hline
		\diagbox{$I$}{$[N,K,D]$}{$r$} & 1 & 3 & 5\\
		\hline
		\{0\} & [81,9,45] & [81,40,9] & [81,70,5]\\
		\hline
		\{2\} & [81,11,36] & [81,39,16] & [81,72,4]\\
		\hline
		\{4\} & [81,5,54] & [81,33,18] & [81,66,6]\\
		\hline
		\{0,2\} & [81,15,27] & [81,48,9] & [81,76,3]\\
		\hline
		\{0,4\} & [81,9,45] & [81,42,9] & [81,70,5]\\
		\hline
		\{2,4\} & [81,11,36] & [81,41,16] & [81,72,4]\\
		\hline
		\{0,2,4\} & [81,15,27] & [81,50,9] & [81,76,3]\\
		\hline
	\end{tabular}
\end{table}
\end{example}

\section{Minimum vectors of sandwiched Reed--Muller codes}\label{sec_minimum_vectors}
In this section, we determine the minimum vectors of $\cC_q(r,I,n)$ for some special cases.
We split our discussion into two subsections as follows.

\subsection{Case (i): $r=\rho(q-1)+1,0\leqslant \rho\leqslant n-1$.}
Throughout this subsection, we assume that $r=\rho(q-1)+1,0\leqslant \rho\leqslant n-1$ and
$$
\begin{cases}
    \{q, q-2, |q-4|\} \cap I=\emptyset & \textup{if $\rho$ is odd},\\
    \{1, 3\} \cap I=\emptyset & \textup{if $\rho$ is even}.
\end{cases}
$$
(Note that this assumption is slightly stronger than that of Theorem~\ref{thm_distance}.)
By Theorem~\ref{thm_distance}, we have 
$$\dist (\cC_q(r,I,n) )= \dist (\cR_q (r-1,n))=q^{n-\rho}.$$
By (a) of Remark~\ref{CIdetails}, the minimum vectors of $\cR_q (r-1,n)$ are also minimum vectors of $\cC_q(r,I,n)$.
Furthermore, we now show that these are exactly the minimum vectors of $\cC_q(r,I,n)$.

\begin{proposition}\label{Locator polynomial}
  Let $r=\rho(q-1)+1,0\leqslant \rho\leqslant n-1$ and assume that
  $$
  \begin{cases}
  	\{q, q-2, |q-4|\} \cap I=\emptyset & \textup{if $\rho$ is odd},\\
  	\{1, 3\} \cap I=\emptyset & \textup{if $\rho$ is even}.
  \end{cases}
  $$
Let $\cC_q=\cC_q(r,I,n)^*$ with length $N=q^{n}-1$.
Let $\delta=q^{n-\rho}-1$, let
$\cI_{n-\rho}=\{q^{n-\rho}-q^j \mid j\in[0,n-\rho-1]\}$
and let $\cS_{\cC_q}(\delta)$ be the newton identities system \eqref{algebraic_system}, written for the codewords of weight $\delta$ of $\cC_q$.
Then, any solution $(\Lambda_1,\cdots,\Lambda_N,\sigma_1,\cdots,\sigma_{\delta})$ of $\cS_{\cC_q}(\delta)$ satisfies the following:
\begin{itemize}
	\item[(i)]
	$\Lambda_1=\Lambda_2=\cdots=\Lambda_{\delta-1}=0.$
	\item[(ii)]
	Let $u\in [1,\delta-1]$.
	\begin{itemize}
		\item[(a)]
		If $u\notin \cI_{n-\rho}$, then $\Lambda_{\delta+u}=0$ and $\sigma_u=0$.
		\item[(b)]
		If $u\in \cI_{n-\rho}$, then $\sigma_u=-\Lambda_{\delta+u}/\Lambda_{\delta}$.
	\end{itemize}
	% If $u\notin \cI_{n-r+1}$, then $\Lambda_{\delta+u}=0$ and $\sigma_u=0$. If $u\in \cI_{n-r+1}$, then $\sigma_u=\Lambda_{\delta+u}/\Lambda_{\delta}$.
	\item[(iii)]
	$$\sigma(X)=1-\sum_{u\in\cI_{n-\rho}} \frac{\Lambda_{\delta+u}}{\Lambda_{\delta}}X^u.$$
\end{itemize}
\end{proposition}

\begin{proof}
  %\begin{itemize}
	(i) Recall that the zero set of $\cC_q(r,I,n)^*$ is
	\begin{align*}
		Z_{r,I} =Z_{r} \bigcup \bigg(\bigcup_{k\in \overline{I}} \Theta ^{(r)}_{k}\bigg).
	\end{align*}
	%A vector $\bfx\in{\mathcal C}(r,I,n)^*$ if and only if $\Lambda_u=0$ for all $u\in Z_{r,I}$.
	Note that $\delta=q^{n-\rho}-1=\sum_{i=0}^{n-\rho-1}(q-1)q^i$, so $\wt_q(\delta)=n(q-1)-r+1$,
	and for any $1\leqslant u\leqslant \delta-1$ we have $\wt_q(u)\leqslant n(q-1)-r$.
	If $\wt_q(u)<n(q-1)-r$, then $\Lambda_u=0$ by definition.
	If $\wt_q(u)=n(q-1)-r$, then it is not hard to check that
	$$
	|O(u)-E(u)| =
	\begin{cases}
		q-1\pm1 & \textup{if $\rho$ is odd},\\
		1 & \textup{if $\rho$ is even},
	\end{cases}
	$$
	whence $\alpha^u\in Z_{r,I}$ and $\Lambda_u=0$. Hence for any $1\leqslant u\leqslant \delta-1$ we always have $\Lambda_u=0$, which proves (i).

	(ii) Set $t=\delta+u$ where $u\in [1,\delta-1]$.
	Then we can write
	\begin{align}\label{t_2adic_expression}
		t=\sum_{i=0}^{n-\rho-1}t_{i}q^{i}+q^{n-\rho},\ t_{i}\in[0,q-1].
	\end{align}
	First we can easily check that $\wt_q(t)\leqslant (n-\rho)(q-1)$ for any $\delta<t<2\delta$.

	Assume that $$\wt_q(t)=(n-\rho)(q-1)=n(q-1)-r+1.$$ Then by~\eqref{t_2adic_expression} and $\delta<t<2\delta$ there is only one $t_{j},1\leqslant j\leqslant n-\rho-1$, equal to $q-2$, i.e. $t=\delta+q^{n-\rho}-q^{j}$, with $j\in [1,n-\rho-1]$. Hence $u=q^{n-\rho}-q^{j}\in \cI_{n-\rho}$.
      Next assume that $$\wt_q(t)=(n-\rho)(q-1)-1=n(q-1)-r.$$ By~\eqref{t_2adic_expression} there exist two cases, one of which is $t_j,t_k,j\ne k,j,k\in [0,n-\rho-1]$, are equal to $q-2$, another is $t_l,l\in [0,n-\rho-1]$, is equal to $q-3$ (when $q\geqslant3$). In either case, it is not hard to verify that
	$$
	|O(t)-E(t)| =
	\begin{cases}
		q,~q-2 \mbox{ or } |q-4|& \textup{if $\rho$ is odd},\\
		1 \mbox{ or } 3& \textup{if $\rho$ is even}.
	\end{cases}
	$$
	Therefore $\alpha^t\in Z_{r,I}$.
	If $\wt_q(\delta+u)\leqslant n(q-1)-r-1$, then $\alpha^{t}\in Z_{r,I}$.
  Combining the above arguments, we conclude that if $u\notin\cI_{n-\rho}$,
  then $\alpha^{\delta+u}\in Z_{r,I}$. Hence $\Lambda_{\delta+u}=0$.

  Now we write the first $2\delta-1$ Newton identities, taking into account the condition $\{\Lambda_i=0,i\in [0,\delta-1]\}$.
	\begin{align*}
		I_1&: \Lambda_{\delta+1}+\Lambda_{\delta}\sigma_1=0 \\
		I_2&: \Lambda_{\delta+2}+\Lambda_{\delta+1}\sigma_1+\Lambda_{\delta}\sigma_2=0\\
		&\cdots\\
		I_u&: \Lambda_{\delta+u}+\sum_{i=1}^{u}\Lambda_{\delta+u-i}\sigma_i=0\\
		&\cdots\\
		I_{2\delta-1}&: \Lambda_{2\delta-1}+\Lambda_{2\delta-2}\sigma_1+\cdots+\Lambda_{\delta}\sigma_{\delta-1}=0.
	\end{align*}
  We now proceed by induction on $u$. Since $1\notin \cI_{n-\rho}$,
  then $\Lambda_{\delta+1}=0$ and $I_1$ gives $\sigma_1=0$ (since $\Lambda_{\delta}\ne 0$).
  Now we assume that for $s\in [1,u-1]$, if $s\notin \cI_{n-\rho}$ then $\sigma_s=0$
  else $\sigma_s=-\Lambda_{\delta+s}/\Lambda_{\delta}$.
  Consider every term $\Lambda_{\delta+u-s}\sigma_s, s<u$, of the identity $I_u$.
	By the assumption we obtain
	\begin{enumerate}
		\item
		If $s\notin \cI_{n-\rho}$ then $\sigma_s=0$.
		\item
		If $s\in \cI_{n-\rho}$ then $s=q^{n-\rho}-q^j$, where $j\in [0,n-\rho-1]$.
    As $s<u<q^{n-\rho}-1$,
    $$u-s<q^{n-\rho}-1-q^{n-\rho}+q^j=q^j-1<(q-1)q^{n-\rho-1}.$$
    Hence, $u-s\notin \cI_{n-\rho}$ implies $\Lambda_{\delta+u-s}=0$.
	\end{enumerate}

	So the identity $I_u$ is in fact $\Lambda_{\delta+u}+\Lambda_{\delta}\sigma_u=0$.
  If $u\notin \cI_{n-\rho}$ then $\Lambda_{\delta+u}=0$ and $\sigma_u=0$,
  otherwise $\sigma_u=-\Lambda_{\delta+u}/\Lambda_{\delta}$, which proves (ii).

	(iii) This follows directly from (ii).
\end{proof}

Now we are ready to show that the minimum vectors of $\cC_q(r,I,n)^{*}$
are exactly those of $\cR_q (r-1,n)^{*}$.

\begin{theorem}\label{thm_min_vec_cI}
  Let $r=\rho(q-1)+1,\ 0\leqslant\rho\leqslant n-1$ and assume that
  $$ 
  \begin{cases}
  	\{q, q-2, |q-4|\} \cap I=\emptyset & \textup{if $\rho$ is odd},\\
  	\{1, 3\} \cap I=\emptyset & \textup{if $\rho$ is even}.
  \end{cases}
  $$
  Then the minimum vectors of $\cC_q(r,I,n)^{*}$ are the codewords
  $$x=\sum_{g\in V_{n-\rho}^{*}}\lambda(g),~\lambda\in \FF_q^*$$
  where $V_{n-\rho}$ is any subspace of $\FF_{q^{n}}$ of dimension $n-\rho$ over $\FF_q$ and $V_{n-\rho}^*=V_{n-\rho}\backslash\{0\}$.
\end{theorem}
\begin{proof}
%  Fix $$I_0=
%  \begin{cases}
%  	M_r\setminus\{q, q-2, |q-4|\} & \textup{if $\rho$ is odd},\\
%  	M_r\setminus\{1, 3\} & \textup{if $\rho$ is even},
%  \end{cases}
%  $$ and 
Suppose that $x$ is a codeword of weight $\delta=q^{n-\rho}-1$ of $\cC_q(r,I,n)^*$.
  By Proposition~\ref{Locator polynomial} we know that its locator polynomial is
  $$\sigma(X)=1-\sum_{u\in\cI_{n-\rho}} \frac{\Lambda_{\delta+u}}{\Lambda_{\delta}}X^u.$$
  Then from Proposition~\ref{THE ROOT}, the roots of $\sigma(X)$ are the nonzero elements of some subspace $V_{n-\rho}$ of $\FF_{q^{n}}$ of dimension $n-\rho$.  
  Therefore 
  $$x=\sum_{g\in V_{n-\rho}^{*}}\lambda_g(g),~\lambda_g\in \FF_q^*.$$
  On the other hand, from Theorem~\ref{RM_min_vectors} we know that the codeword $\sum_{g\in V_{n-\rho}^{*}}(g)\in \cR_q (r-1,n)^{*}$,
  so it is also a minimum vector of $\cC_q(r,I,n)^*$. 
  This implies that $\lambda_g=\lambda$ for some $\lambda\in\FF_q^*$ and all $g\in V_{n-\rho}^{*}$, otherwise this
  would  contradict to the minimum distance of $\cC_q(r,I,n)^*$.
  So
  $$x=\sum_{g\in V_{n-\rho}^{*}}\lambda(g),~\lambda\in \FF_q^*.$$
 % Next let $I\subset M_r$ be any subset such that $$
 % \begin{cases}
 % 	\{q, q-2, |q-4|\} \cap I=\emptyset & \textup{if $\rho$ is odd},\\
 % 	\{1, 3\} \cap I=\emptyset & \textup{if $\rho$ is even}.
 % \end{cases}
 % $$
 % Then $I\subset I_0$, and by (a) and (d) of Remark~\ref{CIdetails} we know that
 % $$\cR_q (r-1,n)^{*} \subseteq \cC_q(r,I,n)^{*} \subseteq \cC_q(r,I_0,n)^{*}$$
 This concludes the proof.
\end{proof}

\begin{theorem}
  Let $r=\rho(q-1)+1,\ 0\leqslant\rho\leqslant n-1$ and assume that
  $$
  \begin{cases}
  	\{q, q-2, |q-4|\} \cap I=\emptyset & \textup{if $\rho$ is odd},\\
  	\{1, 3\} \cap I=\emptyset & \textup{if $\rho$ is even}.
  \end{cases}
  $$
  Then the minimum vectors of ${\cC_q}(r,I,n)$ in the algebra $\mathcal{A}$
  are precisely the codewords
  $$x=\lambda \sum_{g\in V_{n-\rho}}(g+h),~\lambda\in\FF_q^*$$
  where $V_{n-\rho}$ is any subspace of $\FF_{q^{n}}$ of dimension $n-\rho$ over $\FF_q$ and $h\in\FF_{q^{n}}$.

\end{theorem}
\begin{proof}
  Note that ${\cC_q}(r,I,n)$ is affine-invariant by Theorem~\ref{affine_invariant_extended}.
  Combining this with Theorem~\ref{thm_min_vec_cI} we obtain the result.
\end{proof}

\subsection{Case (ii): $r=\rho(q-1)$, $\rho$ is even and $\{0,2\}\cap I=\{0\}$}
In this subsection we assume that $\rho=2\rho'$ is even, so, $r=\rho(q-1)$ is even.

\begin{lemma}\label{min_vec_FF_4}
	Let $n=2m,\ r=\rho(q-1),\ 0\leqslant\rho\leqslant n-1$, $\rho=2\rho'$ and $0\in I$. 
  Let $V_{m-\rho'}^{(2)}$ be a subspace of $\FF_{q^n}$ with dimension $m-\rho'$ over $\FF_{q^2}$.
	Let $x=\sum_{g\in V_{m-\rho'}^{(2)}\backslash\{0\}}\lambda(g),\lambda\in\FF_q^*$, then $x\in\cC_q(r,I,n)^{*}$.
\end{lemma}
\begin{proof}
	By definition we know that the zero set of $\cC_q(r,I,n)^*$ is
	\begin{align*}
		Z_{r,I} =Z_r\bigcup \bigg(\bigcup_{k\in \overline{I}} \Theta ^{(r)}_{k}\bigg).
	\end{align*}
	Now we will prove that $\Lambda_{u}=0$ for all $\alpha^u\in Z_{r,I}$, which implies that
	$x\in\cC_q(r,I,n)^*$.

	Let $\alpha^u\in Z_{r,I}$.
	If $\wt_q(u)\leqslant n(q-1)-r-1=(n-\rho)(q-1)-1$, then by Theorem~\ref{PowerSums} and $V_{m-\rho'}^{(2)}$ is naturally a subspace of dimension $n-\rho$ over $\FF_q$,
	we have $\Lambda_{u}=\sum_{g\in V_{m-\rho'}^{(2)}\backslash\{0\}}g^{u}=0$.
	Now assume that $\wt_q(u)=n(q-1)-r=(n-\rho)(q-1)$ and $O(u)\neq E(u)$. Note that 
  $$\wt_{q^2}(u)=E(u)+qO(u)=\wt_q(u)+(q-1)O(u).$$
	If $E(u)>(m-\rho')(q-1)>O(u)$, we can readily check that $\wt_{q^2}(u)<(m-\rho')(q^2-1)$, then by Theorem~\ref{PowerSums} we have $\Lambda_{u}=0$.
	If $E(u)<(m-\rho')(q-1)<O(u)$, then similarly we have $\wt_{q^2}(qu)<(m-\rho')(q^2-1)$ and $\Lambda_{qu}=\Lambda_{u}^{q}=0$, so $\Lambda_{u}=0$.
	This concludes our proof.
\end{proof}

\begin{lemma}\label{min_vec_FF_4_not}
	Let $n=2m,\ r=\rho(q-1),\ 0\leqslant\rho\leqslant n-1$, $\rho=2\rho'$ and $\{0,2\}\cap I=\{0\}$. 
  Let $V_{n-\rho}$ be a subspace of $\FF_{q^{n}}$ of dimension $n-\rho$ over $\FF_q$
	but not a subspace of dimension $m-\rho'$ over $\FF_{q^2}$.
	Let $x=\sum_{g\in V_{n-\rho}^{*}}\lambda(g)$, $\lambda\in\FF_q^*$, then $x\notin\cC_q(r,I,n)^{*}$.
\end{lemma}
\begin{proof}
	Consider the polynomial
	\[
	f(X)=X^{q^{n-\rho}}\sigma(X^{-1})=X^{q^{n-\rho}}+\sum_{j=0}^{n-\rho-1}\gamma_{j} X^{q^j}
	\]
	where $\sigma(X)$ is the locator polynomial of $x$.
	As $V_{n-\rho}$ is not a subspace of dimension $m-\rho'$ over $\FF_{q^2}$,
	there must exist some odd  $j\in\{1,3,5,\cdots,n-\rho-1\}$ such that $\gamma_{j}\ne 0$.
	On the other hand, note that $\sigma_{q^{n-\rho}-q^{j}}=\gamma_{j}\ne 0$.

    Now by Theorem~\ref{RM_min_vectors}, we get
	$\sigma_{q^{n-\rho}-q^{j}}=-\Lambda_{\delta+q^{n-\rho}-q^{j}}/\Lambda_{\delta}\ne 0$.
	As $\Lambda_{\delta}\ne 0$, we have $\Lambda_{\delta+q^{n-\rho}-q^{j}}\ne 0$, where $\delta=q^{n-\rho}-1$.
	We can readily check that $\wt_q(\delta+q^{n-\rho}-q^{j})=(n-\rho)(q-1)$,
	and $|O(\delta+q^{n-\rho}-q^{j})-E(\delta+q^{n-\rho}-q^{j})|=2$.
	As $2\notin I$, we have $\delta+q^{n-\rho}-q^{j}\in Z_{r,I}$
	and $\Lambda_{\delta+q^{n-\rho}-q^{j}}=0$, which is a contradiction.
	Therefore $x\notin \cC_q(r,I,n)^{*}$.
\end{proof}

\begin{theorem}\label{thm_min_vectors}
	Let $n=2m,\ r=\rho(q-1),\ 0\leqslant\rho\leqslant n-1$, $\rho=2\rho'$ and $\{0,2\}\cap I=\{0\}$, 
  the minimum vectors of ${\cC_q}(r,I,n)^{*}$ are the codewords
	$$x=\sum_{g\in V_{m-\rho'}^{(2)}\backslash\{0\}}\lambda(g),~\lambda\in \FF_q^*$$
	where $V_{m-\rho'}^{(2)}$ is any subspace of $\FF_{q^{n}}$ of dimension $m-\rho'$ over $\FF_{q^2}$.
\end{theorem}
\begin{proof}
	By Theorem~\ref{thm_distance} and Lemma~\ref{min_vec_FF_4}
	we have $$\dist({\cC_q}(r,I,n)^{*})=\dist(\cR_q(r,n)^*)= q^{2(m-\rho')}-1.$$
	So the minimum vectors of ${\cC_q}(r,I,n)^{*}$ are that of $\cR_q(r,n)^*$.
	Now by Lemma~\ref{min_vec_FF_4_not}, we can get that the minimum vectors of ${\cC_q}(r,I,n)^{*}$ are the codewords
	$$x=\sum_{g\in V_{m-\rho'}^{(2)}\backslash\{0\}}\lambda(g),~\lambda\in\FF_q^*$$
	where $V_{m-\rho'}^{(2)}$ is any subspace of $\FF_{q^n}$ of dimension $m-\rho'$ over $\FF_{q^2}$.
	This implies the result.
\end{proof}

\begin{theorem}
	Let $n=2m,\ r=\rho(q-1),\ 0\leqslant\rho\leqslant n-1$, $\rho=2\rho'$ and $\{0,2\}\cap I=\{0\}$. The minimum vectors of ${\cC_q}(r,I,n)$
	in the algebra $\mathcal{A}$ are precisely the codewords
	$$x=\lambda \sum_{g\in V_{m-\rho'}^{(2)}}(g+h),~h\in \FF_{q^n},~\lambda\in\FF_q^*$$
	where $V_{m-\rho'}^{(2)}$ is any subspace of $\FF_{q^{n}}$ of dimension $m-\rho'$ over $\FF_{q^2}$.
\end{theorem}
\begin{proof}
	We know that ${\cC_q}(r,I,n)$ is affine-invariant by Theorem~\ref{affine_invariant_extended}.
	Combining this with Theorem~\ref{thm_min_vectors} we obtain the result.
\end{proof}

\section{Concluding remarks}\label{sec_conclu}
In this paper we propose new families of extended cyclic codes sandwiched between generalized RM codes,
and explicitly determine their minimum vectors for some special cases.
Kudekar et al.~\cite{KKMPSU2017} proved that RM codes achieve capacity
on the BEC both under bit-MAP and block-MAP decoding, thus solving a long standing conjecture. Very recently, Reeves and Pfister~\cite{reeves2021reedmuller} prove that the family of binary RM codes 
achieves capacity on any BMS channel with respect to bit-error rate. 
As our codes are sandwiched between RM codes, it just follows directly that
same results also hold for our new codes, i.e., achieve capacity
on the BEC both under bit-MAP and block-MAP decoding, and achieve capacity on any BMS channel with respect to bit-error rate.
One intriguing problem is to find an $n$-variable polynomial representation for our new sandwiched RM codes. This will help us determine the minimum vectors of those universally strongly perfect lattices of~\cite{hu2020strongly}.

\bibliographystyle{IEEEtran}
\bibliography{ref}

\begin{thebibliography}{10}
\providecommand{\url}[1]{#1}
\csname url@samestyle\endcsname
\providecommand{\newblock}{\relax}
\providecommand{\bibinfo}[2]{#2}
\providecommand{\BIBentrySTDinterwordspacing}{\spaceskip=0pt\relax}
\providecommand{\BIBentryALTinterwordstretchfactor}{4}
\providecommand{\BIBentryALTinterwordspacing}{\spaceskip=\fontdimen2\font plus
\BIBentryALTinterwordstretchfactor\fontdimen3\font minus
  \fontdimen4\font\relax}
\providecommand{\BIBforeignlanguage}[2]{{%
\expandafter\ifx\csname l@#1\endcsname\relax
\typeout{** WARNING: IEEEtran.bst: No hyphenation pattern has been}%
\typeout{** loaded for the language `#1'. Using the pattern for}%
\typeout{** the default language instead.}%
\else
\language=\csname l@#1\endcsname
\fi
#2}}
\providecommand{\BIBdecl}{\relax}
\BIBdecl

\bibitem{muller1954application}
D.~E. Muller, ``Application of boolean algebra to switching circuit design and
  to error detection,'' \emph{Transactions of the IRE Professional Group on
  Electronic Computers}, no.~3, pp. 6--12, 1954.

\bibitem{reed1954textordfemininea}
I.~Reed, ``A class of multiple-error-correcting codes and the decoding
  scheme,'' \emph{Transactions of the IRE Professional Group on Information
  Theory}, vol.~4, no.~4, pp. 38--49, 1954.

\bibitem{MR4190995}
E.~Abbe and M.~Ye, ``Reed-{M}uller codes polarize,'' \emph{IEEE Transactions on
  Information Theory}, vol.~66, no.~12, pp. 7311--7332, 2020.

\bibitem{MR3666961}
S.~Kudekar, S.~Kumar, M.~Mondelli, H.~D. Pfister, E.~\c{S}a\c{s}o\v{g}lu, and
  R.~L. Urbanke, ``{R}eed-{M}uller codes achieve capacity on erasure
  channels,'' \emph{IEEE Transactions on Information Theory}, vol.~63, no.~7,
  pp. 4298--4316, 2017.

\bibitem{massey1992deep}
J.~L. Massey, ``Deep-space communications and coding: A marriage made in
  heaven,'' in \emph{Advanced Methods for Satellite and Deep Space
  Communications}.\hskip 1em plus 0.5em minus 0.4em\relax Springer, 1992, pp.
  1--17.

\bibitem{arikan2010survey}
E.~Arikan, ``A survey of {R}eed-{M}uller codes from polar coding perspective,''
  in \emph{2010 IEEE Information Theory Workshop on Information Theory (ITW
  2010, Cairo)}.\hskip 1em plus 0.5em minus 0.4em\relax IEEE, 2010, pp. 1--5.

\bibitem{MR3400278}
E.~Abbe, A.~Shpilka, and A.~Wigderson, ``Reed-{M}uller codes for random
  erasures and errors,'' \emph{IEEE Transactions on Information Theory},
  vol.~61, no.~10, pp. 5229--5252, 2015.

\bibitem{abbe2020reed}
E.~Abbe, A.~Shpilka, and M.~Ye, ``{R}eed--{M}uller codes: Theory and
  algorithms,'' \emph{IEEE Transactions on Information Theory}, vol.~67, no.~6,
  pp. 3251--3277, 2020.

\bibitem{kasami1968new}
T.~Kasami, S.~Lin, and W.~Peterson, ``New generalizations of the
  {R}eed-{M}uller codes--i: Primitive codes,'' \emph{IEEE Transactions on
  Information Theory}, vol.~14, no.~2, pp. 189--199, 1968.

\bibitem{delsarte1970generalized}
P.~Delsarte, J.-M. Goethals, and F.~J. MacWilliams, ``On generalized
  {R}eed-{M}uller codes and their relatives,'' \emph{Information and Control},
  vol.~16, no.~5, pp. 403--442, 1970.

\bibitem{charpin1998open}
P.~Charpin, ``Open problems on cyclic codes,'' \emph{Handbook of coding
  theory}, vol.~1, pp. 963--1063, 1998.

\bibitem{leducq2012new}
E.~Leducq, ``A new proof of {D}elsarte, {G}oethals and {M}ac{W}illiams theorem
  on minimal weight codewords of generalized {R}eed--{M}uller codes,''
  \emph{Finite Fields and Their Applications}, vol.~18, no.~3, pp. 581--586,
  2012.

\bibitem{BarnesWall}
E.~Barnes and G.~Wall, ``Some extreme forms defined in terms of abelian
  groups,'' \emph{Journal of the Australian Mathematical Society}, vol.~1,
  no.~1, pp. 47--63, 1959.

\bibitem{BroueEnguehard}
M.~Brou{\'e} and M.~Enguehard, ``Une famille infinie de formes quadratiques
  enti{\`e}res; leurs groupes d'automorphismes,'' in \emph{Annales
  scientifiques de l'{\'E}cole Normale Sup{\'e}rieure}, vol.~6, no.~1, 1973,
  pp. 17--51.

\bibitem{conway2013sphere}
J.~H. Conway and N.~J.~A. Sloane, \emph{Sphere packings, lattices and
  groups}.\hskip 1em plus 0.5em minus 0.4em\relax Springer Science \& Business
  Media, 2013, vol. 290.

\bibitem{hu2020strongly}
S.~Hu and G.~Nebe, ``Strongly perfect lattices sandwiched between
  {B}arnes--{W}all lattices,'' \emph{Journal of the London Mathematical
  Society}, vol. 101, no.~3, pp. 1068--1089, 2020.

\bibitem{augot1993etude}
D.~Augot, ``{\'E}tude alg{\`e}brique des mots de poids minimum des codes
  cycliques, m{\'e}thodes d'alg{\`e}bre lin{\'e}aire sur les corps finis.''
  Ph.D. dissertation, Universit{\'e} Pierre et Marie Curie-Paris VI, 1993.

\bibitem{augot1996description}
------, ``Description of minimum weight codewords of cyclic codes by algebraic
  systems,'' \emph{Finite Fields and Their Applications}, vol.~2, no.~2, pp.
  138--152, 1996.

\bibitem{mattson1961new}
H.~Mattson and G.~Solomon, ``A new treatment of {B}ose-{C}haudhuri codes,''
  \emph{Journal of the Society for Industrial and Applied Mathematics}, vol.~9,
  no.~4, pp. 654--669, 1961.

\bibitem{kasami1966weight}
T.~Kasami, ``Weight distributions of {B}ose-{C}haudhuri-{H}ocquenghem codes,''
  \emph{Coordinated Science Laboratory Report no. R-317}, 1966.

\bibitem{kasami1967some}
T.~Kasami, S.~Lin, and W.~W. Peterson, ``Some results on cyclic codes which are
  invariant under the affine group and their applications,'' \emph{Information
  and Control}, vol.~11, no. 5-6, pp. 475--496, 1967.

\bibitem{macwilliams1977theory}
F.~J. MacWilliams and N.~J.~A. Sloane, \emph{The theory of error correcting
  codes}.\hskip 1em plus 0.5em minus 0.4em\relax Elsevier, 1977, vol.~16.

\bibitem{assmus1998polynomial}
E.~Assmus~Jr and J.~Key, ``Polynomial codes and finite geometries,''
  \emph{Handbook of coding theory}, vol.~2, no. part 2, pp. 1269--1343, 1998.

\bibitem{KKMPSU2017}
S.~{Kudekar}, S.~{Kumar}, M.~{Mondelli}, H.~D. {Pfister},
  E.~{\c{S}a\c{s}o\u{g}lu}, and R.~L. {Urbanke}, ``{R}eed-{M}uller codes
  achieve capacity on erasure channels,'' \emph{IEEE Transactions on
  Information Theory}, vol.~63, no.~7, pp. 4298--4316, 2017.

\bibitem{reeves2021reedmuller}
G.~Reeves and H.~D. Pfister, ``{R}eed-{M}uller codes achieve capacity on {BMS}
  channels,'' arXiv:2110.14631, 2021.

\end{thebibliography}
\end{document}